\documentclass[runningheads,a4paper]{llncs}

\usepackage{amssymb}
\usepackage{graphicx}
\usepackage{epsfig}

\usepackage[cmex10]{amsmath}
\usepackage{citesort}
\usepackage{amsfonts}
\usepackage{amssymb}
\usepackage{bbm}
\usepackage{dsfont}
\usepackage{times}
\usepackage{latexsym}
\usepackage{graphicx}
\usepackage{amsfonts}
\usepackage{psfrag}
\usepackage{subfigure}
\usepackage{multicol}

\usepackage{url}
\urldef{\mailsa}\path|{mgyr, gabry, mikkov, lkra, skoglund}@kth.se|
\newcommand{\keywords}[1]{\par\addvspace\baselineskip
\noindent\keywordname\enspace\ignorespaces#1}

\begin{document}

\mainmatter  

\title{On the Transmit Beamforming for MIMO Wiretap Channels: Large-System Analysis}

\titlerunning{On the Transmit Beamforming for MIMO Wiretap Channels}

%
%
\author{Maksym A. Girnyk
\and Fr\'{e}d\'{e}ric Gabry%
\thanks{The present research was supported by the Swedish Research Council (VR).}%
 \and Mikko Vehkaper\"{a}\and Lars K. Rasmussen \and Mikael Skoglund}
\authorrunning{On the Transmit Beamforming for MIMO Wiretap Channels}

\institute{Communication Theory Laboratory, KTH Royal Institute of Technology\\
Osquldas v\"{a}g 10, 10044 Stockholm, Sweden\\
\mailsa\\
\url{http://commth.ee.kth.se}}

%
%

\toctitle{Lecture Notes in Computer Science}
\tocauthor{ICITS-2013}
\maketitle

\newcommand{\cf}{\textit{cf.}}
\newcommand{\eg}{\textit{e.g.}}
\newcommand{\etal}{\textit{et al.}}
\newcommand{\etc}{\textit{etc.}}
\newcommand{\ie}{\textit{i.e.}}
\newcommand{\nb}{\textit{n.b.}}
\newcommand{\perse}{\textit{per se}}
\newcommand{\via}{\textit{via }}
\newcommand{\viceversa}{\textit{vice versa}}
\newcommand{\vide}{\textit{vide }}
\newcommand{\viz}{\textit{viz.}}
\newcommand{\vs}{\textit{vs.}\;}

\newcommand{\eps}{\varepsilon}
\newcommand{\E}{\textsf{E}}
\newcommand{\T}{\textsf{T}}
\renewcommand{\H}{\textsf{H}}
\renewcommand{\l}{\ell}
\newcommand{\til}[1]{\tilde{#1}}
\newcommand{\prim}[1]{#1'}
\newcommand{\diag}{\text{diag}}
\newcommand{\re}{\textrm{Re}}
\newcommand{\im}{\textrm{Im}}
\newcommand{\tr}{\textnormal{\textrm{tr}}}
\newcommand{\writeunder}[2]{\underset{#1}{\underbrace{#2}}}
\newcommand{\round}[1]{\left\lfloor #1 \right\rceil}

\newcommand{\brc}[1]{\left( #1 \right)}
\newcommand{\sqbrc}[1]{\left[ #1 \right]}
\newcommand{\figbrc}[1]{\left\{ #1 \right\} }
\newcommand{\coef}[2]{\textnormal{coef}\left\{#1,#2\right\}}
\newcommand{\ul}[1]{\underline{#1}}
\newcommand{\vv}[1]{\mathbf{#1}}
\newcommand{\bb}[1]{\mathbb{#1}}

\newcommand{\xbl}[1]{\boldsymbol{#1}}
\newcommand{\bx}{\boldsymbol{x}}
\newcommand{\by}{\boldsymbol{y}}
\newcommand{\bs}{\boldsymbol{s}}
\newcommand{\bz}{\boldsymbol{z}}
\newcommand{\bn}{\boldsymbol{n}}
\newcommand{\bv}{\boldsymbol{v}}
\newcommand{\bh}{\boldsymbol{h}}
\newcommand{\qtil}{\tilde{q}}
\newcommand{\ptil}{\tilde{p}}
\newcommand{\bH}{\boldsymbol{H}}
\newcommand{\bA}{\boldsymbol{A}}
\newcommand{\bB}{\boldsymbol{B}}
\newcommand{\bX}{\boldsymbol{X}}
\newcommand{\bN}{\boldsymbol{N}}
\newcommand{\bV}{\boldsymbol{V}}
\newcommand{\bQ}{\boldsymbol{Q}}
\newcommand{\bK}{\boldsymbol{K}}
\newcommand{\bI}{\mathbf{I}}
\newcommand{\bP}{\boldsymbol{P}}
\newcommand{\bone}{\boldsymbol{1}}
\newcommand{\boneT}{\boldsymbol{1}^{\T}}
\newcommand{\bzero}{\boldsymbol{0}}
\newcommand{\bG}{\boldsymbol{G}}
\newcommand{\bSigma}{\boldsymbol{\Sigma}}
\newcommand{\calH}{\mathcal{H}}
\def\calC{{\mathcal{C}}}
\newcommand{\calN}{\mathcal{N}}
\newcommand{\calF}{\mathcal{F}}
\newcommand{\calFbar}{\bar{\mathcal{F}}}
\newcommand{\Zbar}{\bar{Z}}
\newcommand{\calQ}{\mathcal{Q}}
\newcommand{\const}{\mathrm{const}}
\newcommand{\refeqn}[1]{(\ref{#1})}
\newcommand{\epsb}{\bar{\varepsilon}}
\newcommand{\xib}{\bar{\xi}}
\renewcommand{\d}{\; \mathrm{d}}

\newcommand{\main}{\textrm{M}}
\newcommand{\eave}{\textrm{E}}
\newcommand{\Rs}{R_{\textrm{s}}}
\newcommand{\rs}{r_{\textrm{s}}}
\newcommand{\bPwf}{\bP_{\mathrm{WF}}^{\star}}
\newcommand{\bPgsvd}{\bP_{\mathrm{GSVD}}^{\star}}

\newcommand{\betam}{\beta_{\main}}
\newcommand{\betae}{\beta_{\eave}}
\newcommand{\Em}{e_{\main}}
\newcommand{\Ee}{e_{\eave}}
\newcommand{\deltam}{\delta_{\main}}
\newcommand{\deltae}{\delta_{\eave}}
\renewcommand{\Im}{I_{\main}}
\newcommand{\Ie}{I_{\eave}}

\newcommand{\bU}{\boldsymbol{U}}
\newcommand{\bUm}{\boldsymbol{U}_{\main}}
\newcommand{\bUe}{\boldsymbol{U}_{\eave}}
\newcommand{\bVm}{\boldsymbol{V}_{\main}}
\newcommand{\bVe}{\boldsymbol{V}_{\eave}}
\newcommand{\bSigmam}{\boldsymbol{\Sigma}_{\main}}
\newcommand{\bSigmae}{\boldsymbol{\Sigma}_{\eave}}
\newcommand{\bsigmam}[1]{\sigma_{\main,#1}}
\newcommand{\bsigmae}[1]{\sigma_{\eave,#1}}
\newcommand{\sign}{\mathrm{sign}}

\newcommand\norm[1]{\|#1\|}
\newcommand\mmse[1]{\langle#1\rangle}
\newcommand\Complex[1]{{\mathds{C}}^{#1}}
\newcommand{\bAm}{\boldsymbol{A}_{\main}}
\newcommand{\bAe}{\boldsymbol{A}_{\eave}}
\newcommand{\bGm}{\boldsymbol{G}_{\main}}
\newcommand{\bGe}{\boldsymbol{G}_{\eave}}
\newcommand{\bGmH}{\boldsymbol{G}_{\main}^{\H}}
\newcommand{\bGeH}{\boldsymbol{G}_{\eave}^{\H}}
\newcommand{\bHm}{\boldsymbol{H}_{\main}}
\newcommand{\bHe}{\boldsymbol{H}_{\eave}}
\newcommand{\bHmH}{\boldsymbol{H}_{\main}^{\H}}
\newcommand{\bHeH}{\boldsymbol{H}_{\eave}^{\H}}
\newcommand{\bRm}{\boldsymbol{R}_{\main}}
\newcommand{\bRe}{\boldsymbol{R}_{\eave}}
\newcommand{\bTm}{\boldsymbol{T}_{\main}}
\newcommand{\bTe}{\boldsymbol{T}_{\eave}}
\newcommand{\bT}{\boldsymbol{T}}
\newcommand{\bWm}{\boldsymbol{W}_{\main}}
\newcommand{\bWe}{\boldsymbol{W}_{\eave}}
\newcommand{\bQm}{\boldsymbol{Q}_{\main}}
\newcommand{\bQe}{\boldsymbol{Q}_{\eave}}
\newcommand{\bym}{\boldsymbol{y}_{\main}}
\newcommand{\bye}{\boldsymbol{y}_{\eave}}
\newcommand{\bnm}{\boldsymbol{n}_{\main}}
\newcommand{\bne}{\boldsymbol{n}_{\eave}}
\newcommand{\bxm}{\boldsymbol{x}_{\main}}
\newcommand{\bxe}{\boldsymbol{x}_{\eave}}
\newcommand{\bzm}{\boldsymbol{z}_{\main}}
\newcommand{\bze}{\boldsymbol{z}_{\eave}}
\newcommand{\Nm}{N_{\main}}
\newcommand{\Ne}{N_{\eave}}
\newcommand{\rhom}{\rho_{\main}}
\newcommand{\rhoe}{\rho_{\eave}}
\newcommand{\epsm}{\eps_{\main}}
\newcommand{\epse}{\eps_{\eave}}
\newcommand{\xim}{\xi_{\main}}
\newcommand{\xie}{\xi_{\eave}}

\newcommand{\bxMmse}{\langle\boldsymbol{x}\rangle}
\newcommand{\bxa}[1]{\boldsymbol{x}^{(#1)}}
\newcommand{\bna}[1]{\boldsymbol{n}^{(#1)}}
\newcommand{\bva}[1]{\boldsymbol{v}^{(#1)}}
\newcommand{\bvaT}[1]{\boldsymbol{v}^{(#1) \T}}
\newcommand{\bvaH}[1]{\boldsymbol{v}^{(#1) \H}}
\newcommand{\bnaH}[1]{\boldsymbol{n}^{(#1) \H}}
\newcommand{\bxaH}[1]{\boldsymbol{x}^{(#1) \H}}
\renewcommand{\bz}{\boldsymbol{z}}
\newcommand{\bzH}{\boldsymbol{z}^{ \H}}
\renewcommand{\bx}{\boldsymbol{x}}
\newcommand{\bxmmse}{\langle\boldsymbol{x}\rangle}
\newcommand{\bxH}{\boldsymbol{x}^{ \H}}
\newcommand{\eqover}[1]{\overset{(#1)}{=}}

\newcommand{\complex}[1]{\mathds{C}^{#1}}
\newcommand{\real}[1]{\mathds{R}^{#1}}

\newcommand{\bR}{\boldsymbol{R}}
\newcommand{\bW}{\boldsymbol{W}}

\newcommand{\calV}{\mathcal{V}}

\begin{abstract}
With the growth of wireless networks, security has become a fundamental issue in wireless communications due to the broadcast nature of these networks. In this work, we consider MIMO wiretap channels in a fast fading environment, for which the overall performance is characterized by the ergodic MIMO secrecy rate. Unfortunately, the direct solution to finding ergodic secrecy rates is prohibitive due to the expectations in the rates expressions in this setting. To overcome this difficulty, we invoke the large-system assumption, which allows a deterministic approximation to the ergodic mutual information. Leveraging results from random matrix theory, we are able to characterize the achievable ergodic secrecy rates. Based on this characterization, we address the problem of covariance optimization at the transmitter. Our numerical results demonstrate a good match between the large-system approximation and the actual simulated secrecy rates, as well as some interesting features of the precoder optimization.
\end{abstract}

\section{Introduction}\label{sec:intro}
Wireless networks have developed considerably over the last few decades. As
a consequence of the broadcast nature of these networks, communications can
potentially be attacked by malicious parties, and therefore, security has taken
a fundamental role in today’s communications.
The notion of physical layer security (or information-theoretic security) was developed by Wyner in his fundamental work in~\cite{wyner1975wire}.
The \emph{wiretap channel}, which is the simplest model to study secrecy in communications, was introduced therein, consisting
of a transmitter and two communication channels: to a legitimate receiver and to an eavesdropper.

The \emph{secrecy capacity} of the wiretap channel is then defined as the maximum
transmission rate from the transmitter to the receiver, provided
that the eavesdropper does not get any information. Finding the
aforementioned secrecy capacity is a difficult problem in
general, due to its non-convex nature.

Notwithstanding, multiple-input multiple-output (MIMO) communications~\cite{foschini1998limits},~\cite{telatar1999capacity}
have become an emerging topic during the last two decades due to their promising capacity gains.
Similar to communication networks without secrecy constraints, the
overall performance for channels with secrecy constraints is limited by the channels’ conditions. In particular,
the legitimate parties need to have some advantage over the eavesdropper
in terms of channel quality to guarantee secure communications. Many techniques
have been proposed to overcome this limitation; one example is the use of
multi-antenna systems, as in~\cite{oggier2011secrecy},~\cite{shafiee2009towards},~\cite{liu2009note} and~\cite{khisti2010mimome},
where the secrecy capacity of the MIMO wiretap channel with multiple eavesdroppers (MIMOME) was characterized.
These results extend to the problem of secret-key agreement over wireless channels, as in~\cite{blochWir} where key-distillation strategies over quasi-static fading channels were investigated,
and~\cite{blochWong} where the secret-key capacity of MIMO ergodic channels was considered.
Finding the precoder matrix achieving the MIMO secrecy capacity has been attempted in~\cite{khisti2010mimome},~\cite{oggier2011secrecy},
however the general form of the optimal covariance matrix remains unknown.
Nevertheless, in certain regimes, the optimal signaling strategies have been derived.
The high SNR case was investigated in~\cite{khisti2010mimome},
while the optimal transmitting scheme at low SNR was found in ~\cite{gur}.
In~\cite{li2010optimal}, the authors characterized the secrecy
capacity for some special cases of channel matrices with certain
rank properties. The special case where the transmitter and
legitimate receiver have two antennas, whereas the eavesdropper
has a single antenna, has been addressed in~\cite{shafiee2009towards}.
More recently, the same problem has been investigated in a computationally efficient way
in~\cite{fakoorian2012optimal} by developing the \emph{generalized
singular value decomposition} (GSVD)-based beamforming at the
transmitter, and deriving the optimal transmit covariance matrix. Optimal signalling in presence of an isotropic eavesdropper has been recently investigated in~\cite{cha}.
In particular the authors in~\cite{cha} found a close-formed expression for the optimal covariance matrix in the isotropic case
as well as lower and upper bounds on the secrecy capacity for the general case.

All the references above considered quasi-static scenario, where the changes in channel gains were slow enough,
so that the transmitter could adapt its radiation pattern to each channel realization.
If, on the contrary, wireless channels are subject to ergodic fading, a codeword spans many fading realizations and
traditional notion of secrecy rate is no longer suitable. Hence, the concept of \emph{ergodic secrecy rate},
proposed in~\cite{liang2008secure} and~\cite{gop08}, has to be used to characterize the performance of the wiretap channel.
In~\cite{khisti2010misome},~\cite{li2011ergodic} and~\cite{vannguyen2011power} the problem of finding achievable
ergodic secrecy rates was addressed for multiple-input single-output (MISO) channels. In the context of MIMO channels,
in~\cite{lin2}, following a previous work in~\cite{lin1},
the authors characterize the secrecy capacity of an uncorrelated MIMOME channel with only statistical channel state information (CSI) at the transmitter
and investigate the optimal input covariance matrix under a total power constraint.

Unfortunately, for general fast-fading MIMOME channels evaluation of ergodic secrecy rates is problematic due to the necessity of
averaging over the channel realizations. Hence, asymptotic approaches based on methods
from \emph{random matrix theory}~\cite{couillet2011random} have been proposed to
circumvent these difficulties. Typically, such techniques assume that the number of
antennas at the transmitter and the receiver tend to infinity at a constant rate. Then,
an explicit expression -- or a \emph{deterministic equivalent} -- of the ergodic mutual
information (MI) is obtained. The expression is then shown to describe well
the behavior of the systems with realistic (finite) numbers of antennas.

In this paper, we make a first step in studying the problem of the ergodic secrecy rate maximization under power constraint
in MIMO wiretap channels. After computing the deterministic equivalents of the two MIMO channels, we address
the problem of the transmit precoder optimization. We further show that despite being
capacity achieving for a point-to-point MIMO channel, the water-filling strategy
becomes a poor choice in the wiretap setting. For instance, under the
assumption that the transmitter performs the GSVD-based beamforming, we derive the
ergodic-secrecy-rate maximizing transmit covariance matrix, which outperforms the water-filling solution.

\section{System Model}\label{sec:sysMod}
\begin{figure}
\centering
  \psfrag{A}[][][1]{\textsf{Alice}}
  \psfrag{B}[][][1]{\textsf{Bob}}
  \psfrag{E}[][][1]{\textsf{Eve}}
  \psfrag{Hm}[][][1]{$\bH_{\main}$}
  \psfrag{He}[][][1]{$\bH_{\eave}$}
  \psfrag{nm}[][][1]{$\bn_{\main}$}
  \psfrag{ne}[][][1]{$\bn_{\eave}$}
\includegraphics[width=9cm]{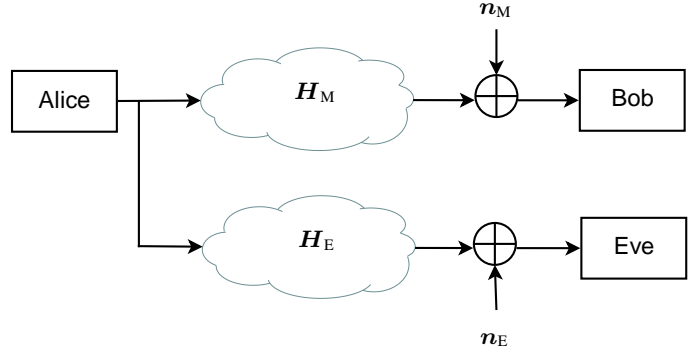}
  \caption{The MIMO wiretap channel.}
  \label{fig:wiretapChannel}
  \vspace{-0.5cm}
\end{figure}

Consider a scenario, where Alice, equipped with an $M$-antenna transmitter, wants to communicate a message to Bob, who is equipped with an $\Nm$-antenna receiver. The message has to be kept secret from unauthorized parties. Meanwhile, Eve tries to eavesdrop the message with the aid of an $\Ne$-antenna receiver. The corresponding setup, depicted in Fig.~\ref{fig:wiretapChannel}, has the following channel model
\begin{subequations} \label{eqn:channelWiretap}
\begin{align}
    \bym =& \bHm \bx + \bnm, \label{eqn:channelMain}\\
    \bye =& \bHe \bx + \bne,   \label{eqn:channelEve}
\end{align}
\end{subequations}
where $\bx \sim \calC\calN(\bzero_{M},\bI_{M})$, $\bnm \sim \calC\calN(\bzero_{\Nm},\bI_{\Nm})$, $\bne \sim \calC\calN(\bzero_{\Ne},\bI_{\Ne})$, and the Kronecker model~\cite{chizhik2000effect} is used, \viz,
\begin{subequations} \label{eqn:kronecker}
\begin{align}
    \bHm =& \sqrt{\frac{\rhom}{M}} \bRm^{1/2} \bWm \bTm^{1/2} \in \Complex{\Nm \times M},\\
    \bHe =& \sqrt{\frac{\rhoe}{M}} \bRe^{1/2} \bWe \bTe^{1/2} \in \Complex{\Ne \times M},
\end{align}
\end{subequations}
where $\bTm$ and $\bRm$ are the transmit and receive correlation matrices of the channel between Alice and Bob, $\bTe$ and $\bRe$ are the transmit and receive correlation matrices of the channel between Alice and Eve, while $\bWm$ and $\bWe$ have i.i.d. $\calC\calN(0,1)$ entries. The channel described by~\eqref{eqn:channelMain} is referred to as the \emph{main channel}, whereas the channel described by~\eqref{eqn:channelEve} is called the \emph{eavesdropper channel}.

For a given transmit covariance matrix, $\bP \triangleq \E \{\bx\bx^{\H}\}$, under the assumption that Alice uses Gaussian signals, the per-antenna achievable ergodic secrecy rate is expressed as
\begin{equation}
    \Rs =  \frac{1}{M} \Bigg[\E_{\bWm} \bigg\{\!\log \det (\bI_{\Nm} + \bHm \bP \bHmH)\!\bigg\}\! - \!\E_{\bWe} \bigg\{\!\log \det (\bI_{\Ne} + \bHe \bP \bHeH)\!\bigg\}\Bigg]^+\!\!\!,  \label{eqn:ergodicSecrecyRate}
\end{equation}
where $[\cdot]^+ = \max\{0,\cdot\}$. Note here the difference to~\cite{fakoorian2012optimal}, where quasi-static fading scenario was considered.

For practical reasons, covariance matrix $\bP$ is assumed to be designed based on the long-term \emph{statistical} CSI, namely, $\{\rhom,\rhoe,\bTm,\bTe,\bRm,\bRe\}$. Note, however, that in order to construct proper wiretap codes, Alice must have access to the \emph{instantaneous} CSI, $\{\bHm,\bHe\}$. Thus, the obtained result is regarded as a computationally efficient lower bound on the achievable secrecy rates.

By choosing the proper covariance matrix $\bP$, one can maximize the achievable secrecy rate of the wiretap channel~\refeqn{eqn:channelWiretap}. The corresponding optimization problem is formulated as

\begin{equation}\label{eqn:optSecrecyCapacity}
\begin{aligned}
    &\underset{\bP}{\max}
    & & \Rs\\
    &\textrm{s.t.}
    & & \tr \{\bP\} \leq M \\
    & & & \bP \succeq \bzero_M.
\end{aligned}
\end{equation}

Unfortunately, the objective function of the above problem has no explicit expression. To evaluate it, one has to perform averaging over the distribution of $\bWm$ and $\bWe$ using, \eg, Monte-Carlo simulation. This approach is, however, quite time-consuming and inefficient. Therefore, a new approach has to be applied to maximize the ergodic secrecy rate. In the following section, we present an asymptotic expression for the ergodic secrecy rate in the limit where dimensions of the channel matrix grow infinitely large.

\section{Achievable Ergodic Secrecy Rate}\label{sec:secrecyRate}
In this section, we provide the large-system approximation for the ergodic secrecy rate of a finite-antenna wiretap channel. We start with the following definition.
\begin{definition}
Given the wiretap channel~\refeqn{eqn:channelWiretap}, the \emph{large-system limit} (LSL) is defined as a regime, where
\begin{align}
    \Nm =&\; \betam M \to \infty,
    & \betam =&\; \const,\\
    \Ne =&\; \betae M \to \infty,
    & \betae =&\; \const.
\end{align}
That is, the numbers of antennas on each side of the channels grow large without bound at constant ratios.
\end{definition}

Based on the above definition, the following proposition presents the large-system approximation for the ergodic MI.
\begin{proposition}\label{thm:lslApproximation}
In the LSL, the following holds
\begin{align}
    \Rs &- \sqbrc{\Im(\rhom) - \Ie(\rhoe)}^+ \to 0,
\end{align}
where
\begin{subequations}\label{eqn:approximationMi}
\begin{align}
    \Im(\rhom) =&\; \frac{1}{M} \log\det\brc{\bI_M + \betam \Em\bTm \bP} + \frac{1}{M} \log\det\brc{\bI_{\Nm} + \deltam\bRm} - \frac{\betam}{\rhom} \deltam \Em \\
    \Ie(\rhoe) =&\; \frac{1}{M} \log\det\brc{\bI_M + \betae \Ee\bTe \bP} + \frac{1}{M} \log\det\brc{\bI_{\Ne} + \deltae\bRe} - \frac{\betae}{\rhoe} \deltae \Ee,
\end{align}
\end{subequations}
and sets of parameters $\{\Em,\deltam\}$ and $\{\Ee,\deltae\}$ form the unique solutions to the following two systems of equations
\begin{subequations}\label{eqn:approximationFpMain}
\begin{align}
    \Em =&\; \frac{\rhom}{\Nm}\tr\figbrc{\bRm\brc{\bI_{\Nm}+\deltam\bRm}^{-1}},\\
    \deltam =&\; \frac{\rhom}{M}\tr\figbrc{\bTm^{1/2}\bP \bTm^{1/2}\brc{\bI_{M}+\betam \Em\bTm^{1/2}\bP \bTm^{1/2}}^{-1}},
\end{align}
\end{subequations}
\vspace{-0.5cm}
\begin{subequations}\label{eqn:approximationFpEve}
\begin{align}
    \Ee =&\; \frac{\rhoe}{\Ne}\tr\figbrc{\bRe\brc{\bI_{\Ne}+\deltae\bRe}^{-1}},\\
    \deltae =&\; \frac{\rhoe}{M}\tr\figbrc{\bTe^{1/2}\bP \bTe^{1/2}\brc{\bI_{M}+\betae \Ee\bTe^{1/2}\bP \bTe^{1/2}}^{-1}},
\end{align}
\end{subequations}
\end{proposition}

\begin{proof}
The proof is based on the concept of a deterministic equivalent~\cite{hachem2007deterministic},~\cite{couillet2011deterministic}. Consider a matrix of the following type
\begin{equation}
    \bB = \bR^{1/2}\bW \bT \bW^{\H} \bR^{1/2},
\end{equation}
where $\bW$ is a random matrix consisting of i.i.d. entries with zero mean and variance $1/M$, while $\bT$ and $\bR$ are Hermitian non-negative definite of bounded normalized trace. The latter are assumed to be generated by tight sequences~\cite{billingsley2008probability}. Moreover, we assume that $\exists \; b>a>0$, such that
\begin{equation}
    a < \lim\inf_N \beta < \lim\sup_N \beta < b,
\end{equation}
where $\beta \triangleq N/M$. As shown in Corollary 1 in~\cite{couillet2011deterministic}, when $N$ and $M$ grow large without bound at ratio $\beta$, the following holds
\begin{equation}
    m(-x) - m^{\circ}(-x) \to 0
\end{equation}
almost surely, where $m(-x)$ is the Stieltjes transform of $\bB$ for $x > 0$ and
\begin{equation}
    m^{\circ}(-x) = \frac{1}{M} \tr \figbrc{\brc{\bI_N + \delta \bR}^{-1}},
\end{equation}
where $e$ and $\delta$ form a unique solution of the following system of fixed-point equations
\begin{subequations}
\begin{align}\label{eqn:eqsTransformStiltjes}
    e =&\; \frac{1}{N} \tr \figbrc{\frac{1}{x}\bR \brc{\bI_N + \delta \bR}^{-1}},\\
    \delta =&\; \frac{1}{M} \tr \figbrc{\frac{1}{x}\bT\brc{\bI_M + \beta e \bT}^{-1}},
\end{align}
\end{subequations}
which, according to Proposition 1 therein, could be solved \via an iterative algorithm always converging to a unique fixed point.

Meanwhile, from Theorem 2 in~\cite{couillet2011deterministic} it follows that under the aforementioned assumptions and some additional constraints on spectral radius of matrices $\bT$ and $\bR$, the Shannon transform~\cite{tulino2004random} of $\bB$ satisfies
\begin{equation}
    \calV(-x) - \calV^{\circ}(-x) \to 0
\end{equation}
almost surely, where
\begin{equation}\label{eqn:transformShannon}
    \calV^{\circ}(-x) = \frac{1}{M} \log\det\brc{\bI_M + \beta e\bT} + \frac{1}{M} \log\det\brc{\bI_{N} + \delta\bR} - x \beta \delta e.
\end{equation}
The above Shannon transform represents the asymptotic behavior of the mean MI in the LSL. Thus, having computed~\eqref{eqn:transformShannon} at $x = 1/\rho$, with parameters satisfying~\eqref{eqn:eqsTransformStiltjes}, we can evaluate the ergodic MI of each MIMO channel within our wiretap model (viz., the main and eavesdropper's channels). To address the influence of the transmit covariance matrix, it suffices to consider $\bT\bP^{1/2}$ instead of $\bT$ for both channels. This leads us exactly to~\eqref{eqn:approximationMi},~\eqref{eqn:approximationFpMain} and~\eqref{eqn:approximationFpEve}, thereby completing the proof.
\end{proof}

\section{Transmit Covariance Optimization}\label{sec:optCovariance}
Based upon the asymptotic analysis carried out in the previous section, here we address the problem of transmit covariance optimization~\refeqn{eqn:optSecrecyCapacity}. As mentioned before, working directly with~\refeqn{eqn:ergodicSecrecyRate} is prohibitive due to expectation operators therein. Moreover, as we have seen from the previous section, the influence of the random parts of the channels $\bWm$ and $\bWe$ vanishes in the LSL. Thus, the objective function of the corresponding optimization problem simplifies to
\begin{equation}\label{eqn:objective}
     \rs(\bP) = \frac{1}{M}\bigg[\log\det\brc{\bI_M + \betam \Em\bTm \bP} - \log\det\brc{\bI_M + \betae \Ee\bTe \bP}\bigg]^+.
\end{equation}
Note that here, we consider $\Em$ and $\Ee$ as independent of the optimization variable $\bP$ due to the following reason. The optimal solution of the optimization problem has to satisfy the KKT conditions, which require that $\nabla_{\bP}\rs(\bP) = \bzero$. When taking into account the dependence of $\Em$ and $\Ee$ on $\bP$, one has to take the derivatives of $\rs(\bP)$ w.r.t. the former. However, it can be verified that those are zero, and hence interdependence between $\Em$, $\Ee$ and $\bP$ does not play any role in the optimization.

Unfortunately, since the problem is non-convex, finding the optimal covariance of $\bx$ is difficult. Hence, we will provide several suboptimal solutions that give a lower bound on the secrecy capacity of the ergodic MIMO wiretap channel.

\subsection{Water-Filling over the Main Channel}
Isotropic transmission is the simplest strategy Alice can perform. However, it is not capacity achieving even for a generic MIMO channel. Instead, based on the statistical CSI of the main channel, $\{\bTm,\bRm\}$, Alice can perform SVD $\sqrt{\betam\Em}\bTm^{1/2} = \bU \bSigma \bV^{\H}$, where $\bU$ and $\bV$ are orthonormal matrices. Then, optimal transmit covariance is given by the \emph{water-filling} (WF) solution as follows
\begin{equation}
    \bPwf = \bV \bSigma_{\bP} \bV^{\H},
\end{equation}
where $[\bSigma_{\bP}]_{m,m} = \sqbrc{\mu^{-1} - [\bSigma]_{m,m}^{-1}}^+$, and $\mu$ is chosen to satisfy the power constraint. In this case Alice acts as if Eve did not exist, achieving the ergodic capacity of the main channel. However, in the presence of an eavesdropper this strategy may be quite inefficient, as we shall see later on.

\subsection{GSVD-Based Precoder}
Consider the scenario where the transmitter performs GSVD on the matrices related to channels~\eqref{eqn:channelMain} and~\eqref{eqn:channelEve}. Although the solution based on this assumption is suboptimal, it is advantageous, as compared to the isotropic precoding. Moreover, it takes into account the presence of the eavesdropper and can potentially increase the ergodic secrecy rate as compared to the WF precoder.

When applied to~\eqref{eqn:objective}, the GSVD-based beamforming method is realized as follows. Based on the statistical CSI of both channels, $\{\bTm,\bRm,\bTe,\bRe\}$, Alice performs GSVD on matrices $\sqrt{\betam\Em}\bTm^{1/2}$ and $\sqrt{\betae\Ee}\bTe^{1/2}$
\begin{align}
     \sqrt{\betam\Em}\bTm^{1/2} =&\; \bUm \bSigmam \bV^{\H},\\
     \sqrt{\betae\Ee}\bTe^{1/2} =&\; \bUe \bSigmae \bV^{\H},
\end{align}
where $\bSigmam^{\T}\bSigmam + \bSigmae^{\T}\bSigmae = \bI_M$. The above GSVD simultaneously diagonalizes $\bTm^{1/2}$ and $\bTe^{1/2}$, converting those into a set of parallel subchannels. Then, the transmitted vector is constructed as $\bx = \bV^{-\H}\bs$, where $\bs\sim\calC\calN(\bzero_M,\bP)$ and $\bP$ is a positive semi-definite diagonal matrix representing the power allocation across the subchannels. For the above beamforming strategy, the optimal power allocation was derived in~\cite{fakoorian2012optimal} (here we have corrected the minor typo therein)
\begin{equation}
    [\bPgsvd]_{i,i} = \frac{1}{2}\sqbrc{\sign(\bsigmam{i}\!-\!\bsigmae{i})\!+\! 1}\!\! \sqbrc{\frac{-\!1\!+\!\sqrt{1\!-\!4\bsigmam{i}\bsigmae{i}\!+\!\frac{4 (\bsigmam{i}-\bsigmae{i})\bsigmam{i}\bsigmae{i}}{\log(2)\mu v_{i}}}}{2\bsigmam{i}\bsigmae{i}}}^+\!\!\!\!,
\end{equation}
where $\bsigmam{i}$, $\bsigmae{i}$ and $v_{i}$ are the $i$th diagonal entries of $\bSigmam^{\T}\bSigmam$, $\bSigmae^{\T}\bSigmae$ and $\bV^{-1}\bV^{-\H}$, respectively, and $\mu$ is chosen to satisfy the power constraint at the transmitter.

\begin{figure}
\centering
\includegraphics[width=9cm]{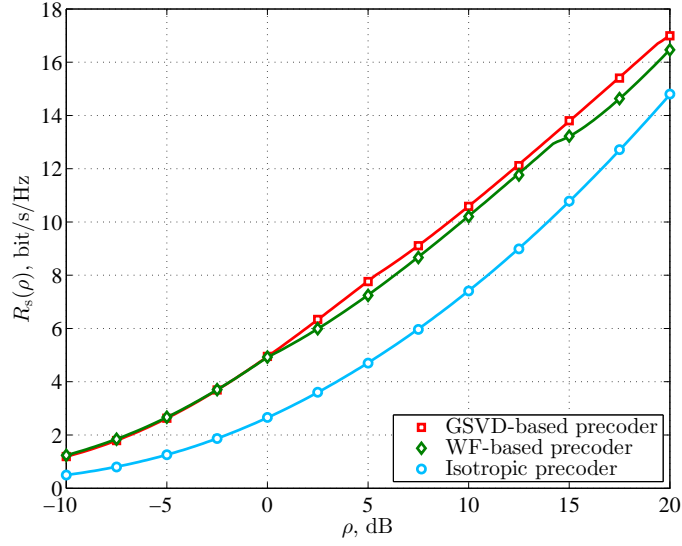}
  \caption{Ergodic secrecy rate \vs SNR ($\rhom = \rhoe = \rho$) for a MIMO wiretap channel with $M = 6$, $\Nm = 6$ and $\Ne = 2$ antennas. Transmit side correlation parameters: $d_{\lambda} = 1$, $\theta_{\main} = 40^{\circ}$, $\theta_{\eave} = -10^{\circ}$, $\Delta_{\main} = \Delta_{\eave} = 5^{\circ}$. Solid curves denote analytic results, while markers denote simulated values averaged over 10 000 channel realizations.}
  \label{fig:resultMiWiretapSnrLsl6x6x2}
  \end{figure}

\section{Numerical Results}\label{sec:numResults}

In this section, we provide results based on numerical simulations along with some discussion. As seen from the objective function~\eqref{eqn:objective}, spatial correlation at the receiver side has no effect on the precoding design. Hence, for the sake of simplicity, we assume that $\bRm = \bI_{\Nm}$ and $\bRe = \bI_{\Ne}$. Meanwhile, correlation at the transmitter side is assumed to be generated by a uniform linear antenna array with \emph{Gaussian power azimuth spectrum}~\cite{wen2007asymptotic}, so that the entries of correlation matrices $\bTm$ and $\bTe$) are obtained by
\begin{equation}\label{eqn:correlation}
    \sqbrc{\bT}_{a,b} = \frac{1}{2 \pi \Delta^2} \int_{-\pi}^{\pi} e^{2 \pi \j d_{\lambda}(a-b)\sin(\phi) - \frac{(\phi-\theta)^2}{2\Delta^2}} \d \phi,
\end{equation}
where $d_{\lambda}$ is the relative antenna spacing (in wavelengths $\lambda$), $\theta$ is the mean angle and $\Delta^2$ is the mean-square angle spread.

\begin{figure}
\centering
\includegraphics[width=9cm]{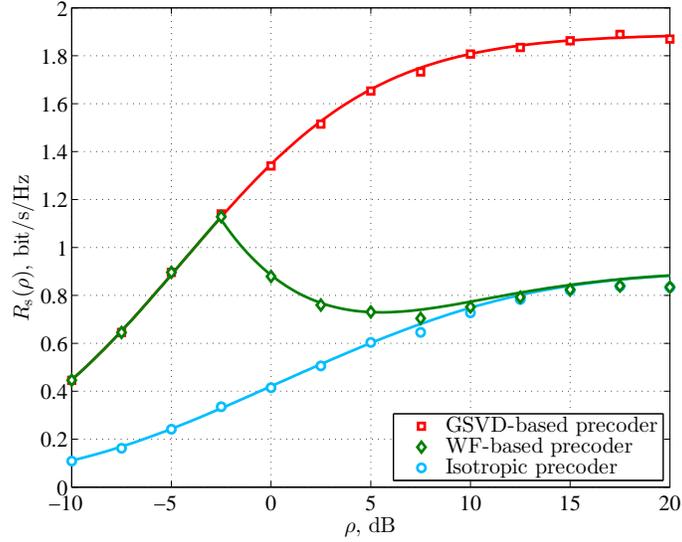}
  \caption{Ergodic secrecy rate \vs SNR ($\rhom = \rhoe = \rho$) for a MIMO wiretap channel with $M = 2$, $\Nm = 3$ and $\Ne = 4$ antennas. Transmit side correlation parameters: $d_{\lambda} = 1$, $\theta_{\main} = 40^{\circ}$, $\theta_{\eave} = -10^{\circ}$, $\delta_{\main} = \Delta_{\eave} = 5^{\circ}$. Solid curves denote analytic results, while markers denote simulated values averaged over 10 000 channel realizations.}
  \label{fig:resultMiWiretapSnrLsl}
\end{figure}

First, we plot in Fig.~\ref{fig:resultMiWiretapSnrLsl6x6x2}, the dependence of the ergodic secrecy rate on the SNR. The transmit side correlation parameters are set as follows. The antenna numbers are set to $M = 6$, $\Nm = 6$ and $\Ne = 2$. The antenna spacing is set to one wavelength, the mean angles are set to $\theta_{\main} = 40^{\circ}$, $\theta_{\eave} = -10^{\circ}$ and the root-mean-square angle spread is chosen for both channels to be $\Delta_{\main} = \Delta_{\eave} = 5^{\circ}$. From the figure, we see that the the results derived in the LSL (solid lines) match the simulations (markers) quite well even for relatively small numbers of antennas. Moreover, we also see that ``statistical" water-filling over the main channel performs well, approaching the performance of the GSVD-based precoding. The isotropic precoder also achieves quite high ergodic secrecy rates, which can be explained by a small number of antennas at the eavesdropper.

Fig.~\ref{fig:resultMiWiretapSnrLsl} depicts similar dependence of the ergodic secrecy rate~\eqref{eqn:ergodicSecrecyRate} on the SNR with different network parameters. The transmit side correlation parameters are chosen similar to the previous case, while the antenna numbers are set to $M = 2$, $\Nm = 3$ and $\Ne = 4$. From the figure we see that water-filling over the main channel is far from being optimal in this case. This can be explained by the fact that in this setting Eve has many antennas and is therefore quite powerful in terms of eavesdropping capabilities. Hence, maximizing the data rate of the main channel, while ignoring the eavesdropper, is a poor strategy in this case. The same observation applies to isotropic precoding, which performs even worse. On the other hand, ``statistical" GSVD-based beamforming proves the most efficient among the considered strategies.

To emphasize the advantage of the GSVD we plot the ergodic secrecy rate as a function of the number of antennas at Eve's receiver, $\Ne$, in Fig.~\ref{fig:resultMiWiretapMLsl}. We fix $d_{\lambda} = 1$ and keep the same parameters as in the previous figure. From Fig.~\ref{fig:resultMiWiretapMLsl} we see that both the isotropic precoding and water-filling cannot provide strictly positive ergodic secrecy rates when $\Ne$ grows large. At the same time we observe that GSVD-based precoding allows to efficiently allocate the power to achieve strictly positive ergodic secrecy rates even when $\Ne$ becomes much larger than $M$ and $\Nm$.

In Fig.~\ref{fig:resultMiWiretapDLsl}, we plot the ergodic secrecy rate $\Rs$ against the spacing between the neighboring antennas within the array. The rest of the transmit-side correlation parameters remain unchanged and the SNR is set to $\rho = 0$ dB. Firstly, we note that the achievable ergodic secrecy rates are non-convex and non-monotone functions of the antenna spacing. Similar behavior was previously observed in~\cite{moustakas2005statistical} and, moreover, the results obtained \via the asymptotic approximation (solid lines) are confirmed with the Monte-Carlo simulation results (markers). Nevertheless, quite interestingly, it can be observed that at low SNR, the optimized secrecy rates are significantly higher than those obtained by the isotropic precoding. Moreover, those are even higher than the secrecy capacity of an uncorrelated wiretap channel, meaning that it can be advantageous to have correlation at low SNR, provided that the transmit covariance is optimized. Finally, we point out that again, as expected, the GSVD-based beamforming reveals to be the most efficient among other choices.

\begin{figure}
\centering
\includegraphics[width=9cm]{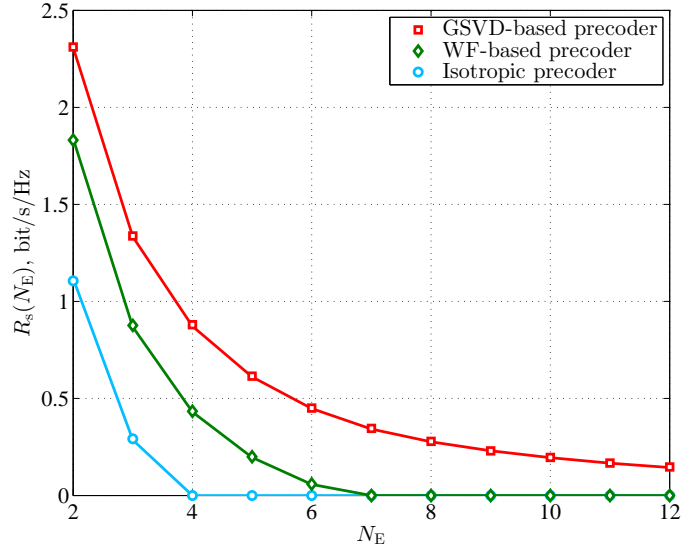}
  \caption{Ergodic secrecy rate \vs number of Eve's antennas $\Ne$ for a MIMO wiretap channel with $M = \Nm = 4$ antennas in the main channel. Transmit side correlation parameters: $d_{\lambda} = 1$, $\theta_{\main} = 40^{\circ}$, $\theta_{\eave} = -10^{\circ}$, $\Delta_{\main} = \Delta_{\eave} = 5^{\circ}$. SNR is set to $\rhom = \rhoe = 0$ dB. Solid curves denote analytic results, while markers denote simulated values averaged over 10 000 channel realizations.}
  \label{fig:resultMiWiretapMLsl}
\end{figure}

\section{Conclusions}\label{sec:conclusion}
In the present paper, we have studied the ergodic secrecy rate of a multi-antenna wiretap channel. Using the theory of deterministic equivalents, we have obtained the large-system approximation of the achievable ergodic secrecy rate, which holds when the numbers of antennas at each terminal grow very large at constant ratios. The approximation proved accurate even for small numbers of antennas, thereby simplifying the computationally demanding problem of transmit covariance optimization. First, not only the objective function of the corresponding optimization problem has closed-form expression, but it has interesting properties attributed to log-det expressions. Secondly, the objective depends only on the correlation matrices of the channels, which can be known at the transmitter by the widely adopted statistical CSI assumption. Once the approximation was obtained, we were able to use some existing algorithms for the covariance optimization. In particular, we have shown that GSVD-based beamforming performs well, compared to, \eg, water-filling over the main channel.

\begin{figure}
\centering
\includegraphics[width=9cm]{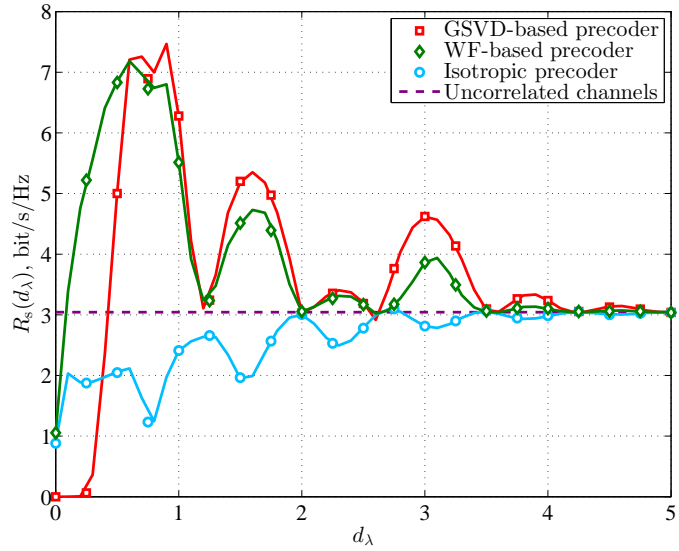}
  \caption{Ergodic secrecy rate \vs antenna spacing $d_{\lambda}$ for a MIMO wiretap channel with $M = 4$, $\Nm = 4$ and $\Ne = 2$ antennas. Transmit side correlation parameters: $\theta_{\main} = 40^{\circ}$, $\theta_{\eave} = -10^{\circ}$, $\Delta_{\main} = \Delta_{\eave} = 5^{\circ}$. SNR is set to $\rhom = \rhoe = 0$ dB. Solid curves denote analytic results, while markers denote simulated values averaged over 10 000 channel realizations.}
  \label{fig:resultMiWiretapDLsl}
\end{figure}

\end{document}